\documentclass[12pt]{article}

\usepackage{hyperref}

\usepackage{amsmath} 
\usepackage{amssymb}
\usepackage{amsfonts}
\usepackage{amsthm}

\newtheorem{lemma}{Lemma}
\newtheorem{corollary}{Corollary}
\newtheorem{proposition}{Proposition}

\theoremstyle{definition}

\newtheorem{definition}{Definition}

\begin{document}

\begin{center}
	\Large
	\textbf{Effective Heisenberg equations for quadratic Hamiltonians}	
	
	\large 
	\textbf{A.E. Teretenkov}\footnote{Department of Mathematical Methods for Quantum Technologies, Steklov Mathematical Institute of Russian Academy of Sciences,
		ul. Gubkina 8, Moscow 119991, Russia\\ E-mail:\href{mailto:taemsu@mail.ru}{taemsu@mail.ru}}
\end{center}

\footnotesize
We discuss effective quantum dynamics obtained by averaging projector with respect to free dynamics. For unitary dynamics generated by quadratic fermionic Hamiltonians we obtain effective Heisenberg dynamics. By perturbative expansions we obtain the correspondent effective time-local Heisenberg equations. We also discuss a similar problem for bosonic case.
\normalsize


\section{Introduction}

In Ref.~\cite{Ter21} we have discussed the effective Gibbs state, which could be obtained from the exact Gibbs state by applying an averaging projector with respect to free dynamics. Here we consider a dynamical analog of such a procedure. Namely, in section \ref{sec:averProj} we define an averaging projector which  acts on dynamical maps and we regard the result as an effective dynamical map. Let us remark that the fact that such an operator acts on  dynamical maps rather than states or observables differs it both from equilibrium situation\cite{Ter21} and from the projectors which are usually discussed in Nakajima–Zwanzig  projection formalism.\cite{Nakajima58,Zwanzig60} In spite of the fact, that very different projection superoperators\cite{Breuer07, Mancal12, Trushechkin19, Semin20} are used in Nakajima–Zwanzig  projection formalism, they usually act on states or observables, so the direct analog of our projector is a ''hyperoperator'' (map from superoperators to superoperators) that acts on a dynamical map by Nakajima–Zwanzig projection superoperators on both sides of the dynamical map. 

Let us also remark that the effective dynamical map commutes with free evolution which is analogous to the dynamical maps for reduced dynamics in the weak coupling limit.\cite{Davies1974,Accardi2002} Usually time-local master equations beyond the weak coupling limit derived by time-convolutionless perturbation theory\cite{Chaturvedi79, Breuer02} with the Argyres-Kelley projection superoperator \cite{Argyres64} and they lack such a commutativity property. But our averaging projector preserves this commutativity property in all the orders of time-convolution perturbation theory. So despite the mentioned differences with Nakajima–Zwanzig  projection superoperators in principle our projection could be used in synergy with it to obtain properties like such commutativity if needed.

In section \ref{sec:effFermHeis} we apply our approach to the fermionic system with s quadratic Hamiltonian. Let us remark that similar to Refs.~\cite{Teretenkov2020Dynamics},\cite{Nosal2020} we are interested in Heisenberg dynamics (effective in our case) of products of creation and annihilation operators rather than density matrix dynamics. In section \ref{sec:timeLocalFermHeis} by time-convolutionless perturbative expansion we obtain time-local effective Heisenberg equations. In section \ref{sec:bosonicCase} we develop similar calculations for the bosonic case at the physical level of rigor, but we show that such a formal analogy leads to some contradictive results. 

In section \ref{sec:Concl} we summarize our results and discuss some open problems and directions for further development.

\section{Averaging projector for dynamical maps}
\label{sec:averProj}
	
	Let us consider a finite-dimensional Hilbert space $ \mathcal{H} $.  We denote the linear operators in $ \mathcal{B}(\mathcal{H}  )$ and superoperators, i.e. linear operators in  $ \mathcal{B}(\mathcal{H}  )$   as  $ \mathcal{B}^2(\mathcal{H}  )$  just as shortcut for $ \mathcal{B}(\mathcal{B}(\mathcal{H}  ))$. Similarly, we denote by  $ \mathcal{B}^3(\mathcal{H}  )$ the linear operators in  $ \mathcal{B}^2(\mathcal{H}  )$. We will use dot-notation for superoperators constructed explicitly from right and left multiplication by operators. E.g., if  $ X, Y \in \mathcal{B}(\mathcal{H}  ) $, then $ X \cdot Y \in  \mathcal{B}^2(\mathcal{H}) $ defined by $ (X \cdot Y) Z = X Z Y  $ for an arbitrary $ Z \in \mathcal{B}(\mathcal{H}  )  $.

Using an analogy with Ref.~\cite{Ter21} we define the averaging projector in the following way.
\begin{definition}
	Let $ \hat{H}_0 \in \mathcal{B}(\mathcal{H}) $ be self-adjoint $ \hat{H}_0^{\dagger} = \hat{H}_0 $  and  $ \Phi \in \mathcal{B}^2(\mathcal{H}) $, then let us define a map	$ \mathfrak{P} \in \mathcal{B}^3(\mathcal{H}) $ by the following formula
	\begin{equation}\label{eq:averProjDef}
		\mathfrak{P}(\Phi) \equiv \lim\limits_{T \rightarrow \infty} \frac{1}{T} \int_0^T ds e^{-i \hat{H}_0 s} \Phi( e^{i  \hat{H}_0 s} \; \cdot \;e^{-i  \hat{H}_0 s}  )  e^{i  \hat{H}_0 s}.
	\end{equation}
\end{definition}

Let us remark that $ \mathfrak{P} $ depends on the choice of $ \hat{H}_0  $, which we call free Hamiltonian, in spite of the fact, that we do not write it explicitly in the notation $ \mathfrak{P} $.  

\begin{proposition}
	Let the spectral decomposition of $ \hat{H}_0  $ have the form
	\begin{equation*}
		\hat{H}_0 = \sum_{\varepsilon} \varepsilon \Pi_{\varepsilon},
	\end{equation*}
	where $ \varepsilon $ are eigenvalues of $ \hat{H}_0 $ and $ \Pi_{\varepsilon} $ are projectors on correspondent eigenspaces $ \Pi_{\varepsilon}^{\dagger} =\Pi_{\varepsilon}$, $\Pi_{\varepsilon} \Pi_{\varepsilon'} = \delta_{\varepsilon \varepsilon'} \Pi_{\varepsilon} $, then
	\begin{equation}\label{eq:averProjRep}
		\mathfrak{P}(\Phi) = \sum_{ \varepsilon_1 - \varepsilon_2 + \varepsilon_3 - \varepsilon_4=0}  \Pi_{\varepsilon_1} \Phi( \Pi_{\varepsilon_2} \; \cdot \;  \Pi_{\varepsilon_3}   )   \Pi_{\varepsilon_4} .
	\end{equation}
\end{proposition}

\begin{proof}
	By Eq. \eqref{eq:averProjDef}
	\begin{equation*}
		\mathfrak{P}(\Phi) = \sum_{\varepsilon_1, \varepsilon_2, \varepsilon_3, \varepsilon_4}\lim\limits_{T \rightarrow \infty} \frac{1}{T} \int_0^T ds e^{-i(\varepsilon_1 - \varepsilon_2 + \varepsilon_3 - \varepsilon_4) s} \Pi_{\varepsilon_1} \Phi( \Pi_{\varepsilon_2} \; \cdot \;  \Pi_{\varepsilon_3})    \Pi_{\varepsilon_4}
	\end{equation*}
	Taking into account
	\begin{equation*}
		\lim\limits_{T \rightarrow \infty} \frac{1}{T} \int_0^T ds e^{-i \varepsilon s} = \delta_{\varepsilon 0} 
	\end{equation*}
	we obtain Eq. \eqref{eq:averProjRep}.
\end{proof}

\begin{proposition}\label{prop:propOfP}
	\begin{enumerate}
		\item The map $ \mathfrak{P} $ is an idempotent
		\begin{equation*}
			\mathfrak{P}^2 = \mathfrak{P}
		\end{equation*}
		\item The superoperator obtained by application of $ \mathfrak{P} $ commutes with superoperator of free evolution
		\begin{equation*}
			[\mathfrak{P}(\Phi), e^{i  \hat{H}_0 t} \; \cdot \; e^{-i  \hat{H}_0 t}] = 0, \qquad \forall \Phi \in \mathcal{B}^2(\mathcal{H}).
		\end{equation*}
		\item Pulling the free dynamics through the map $ \mathfrak{P} $ 
		\begin{equation*}
			\mathfrak{P}( e^{i  \hat{H}_0 t} \Phi( \; \cdot \; ) e^{-i  \hat{H}_0 t}) = e^{i  \hat{H}_0 t}\mathfrak{P}(  \Phi(\; \cdot \; ) ) e^{-i  \hat{H}_0 t},  \quad \forall  \Phi \in \mathcal{B}^2(\mathcal{H}).
		\end{equation*}
	\end{enumerate}
\end{proposition}

\begin{proof}
	\begin{enumerate}
		\item Expanding $ \mathfrak{P} $ by Eq. \eqref{eq:averProjRep} we have
	\begin{eqnarray*}
		&\mathfrak{P}^2(\Phi) = \sum_{ \varepsilon_1 - \varepsilon_2 + \varepsilon_3 - \varepsilon_4=0} \mathfrak{P} (\Pi_{\varepsilon_1} \Phi( \Pi_{\varepsilon_2} \; \cdot \;  \Pi_{\varepsilon_3}   )   \Pi_{\varepsilon_3} ) \\
		&= \sum_{ \varepsilon_1' - \varepsilon_2' + \varepsilon_3' - \varepsilon_4'=0}\sum_{ \varepsilon_1 - \varepsilon_2 + \varepsilon_3 - \varepsilon_4=0} \Pi_{\varepsilon_1'} \Pi_{\varepsilon_1} \Phi( \Pi_{\varepsilon_2} \Pi_{\varepsilon_2'} \; \cdot \;  \Pi_{\varepsilon_3'}  \Pi_{\varepsilon_3}   )   \Pi_{\varepsilon_4} \Pi_{\varepsilon_4'} \\
		&= \sum_{ \varepsilon_1' - \varepsilon_2' + \varepsilon_3' - \varepsilon_4'=0}\sum_{ \varepsilon_1 - \varepsilon_2 + \varepsilon_3 - \varepsilon_4=0} \delta_{\varepsilon_1 \varepsilon_1'} \delta_{\varepsilon_2 \varepsilon_2'} \delta_{\varepsilon_3 \varepsilon_3'} \delta_{\varepsilon_4 \varepsilon_4'}\Pi_{\varepsilon_1} \Phi( \Pi_{\varepsilon_2} \; \cdot \;  \Pi_{\varepsilon_3}   )   \Pi_{\varepsilon_4} \\
		&= \sum_{ \varepsilon_1 - \varepsilon_2 + \varepsilon_3 - \varepsilon_4=0}  \Pi_{\varepsilon_1} \Phi( \Pi_{\varepsilon_2} \; \cdot \;  \Pi_{\varepsilon_3}   )   \Pi_{\varepsilon_4}= \mathfrak{P}(\Phi)
	\end{eqnarray*}

	\item  Taking into account Eq. \eqref{eq:averProjRep} once again, we obtain
	\begin{eqnarray*}
		&\mathfrak{P}(\Phi(  e^{i  \hat{H}_0 t} \; \cdot \;e^{-i  \hat{H}_0 t})) = \sum_{ \varepsilon_1 - \varepsilon_2 + \varepsilon_3 - \varepsilon_4=0}  \Pi_{\varepsilon_1} \Phi( \Pi_{\varepsilon_2} e^{i  \hat{H}_0 t} \; \cdot \;e^{-i  \hat{H}_0 t}  \Pi_{\varepsilon_3}   )   \Pi_{\varepsilon_4} \\
		&=  \sum_{ \varepsilon_1 - \varepsilon_2 + \varepsilon_3 - \varepsilon_4=0} e^{i (\varepsilon_2 - \varepsilon_3)t}\Pi_{\varepsilon_1} \Phi( \Pi_{\varepsilon_2}\; \cdot \  \Pi_{\varepsilon_3}   )   \Pi_{\varepsilon_4} \\
		&=  \sum_{ \varepsilon_1 - \varepsilon_2 + \varepsilon_3 - \varepsilon_4=0} e^{i (\varepsilon_1 - \varepsilon_4)t}\Pi_{\varepsilon_1} \Phi( \Pi_{\varepsilon_2}\; \cdot \  \Pi_{\varepsilon_3}   )   \Pi_{\varepsilon_4} \\
		&=  \sum_{ \varepsilon_1 - \varepsilon_2 + \varepsilon_3 - \varepsilon_4=0} \Pi_{\varepsilon_1}  e^{i  \hat{H}_0 t}\Phi( \Pi_{\varepsilon_2}\; \cdot \  \Pi_{\varepsilon_3}   )   \Pi_{\varepsilon_4}   e^{-i  \hat{H}_0 t} \\
		&=e^{i  \hat{H}_0 t} \mathfrak{P}(\Phi) e^{-i  \hat{H}_0 t}.
	\end{eqnarray*}

	\item By Eq. \eqref{eq:averProjDef}
	\begin{eqnarray*}
		&	\mathfrak{P}( e^{i  \hat{H}_0 t} \Phi( \; \cdot \; ) e^{-i  \hat{H}_0 t}) = \lim\limits_{T \rightarrow \infty} \frac{1}{T} \int_0^T ds e^{-i \hat{H}_0 s} e^{i  \hat{H}_0 t} \Phi(\; \cdot \; ) e^{-i  \hat{H}_0 t}   e^{i  \hat{H}_0 s} \\
		& =  e^{i  \hat{H}_0 t} \left( \lim\limits_{T \rightarrow \infty} \frac{1}{T} \int_0^T ds e^{-i \hat{H}_0 s} \Phi(\; \cdot \; )    e^{i  \hat{H}_0 s} \right)e^{-i  \hat{H}_0 t}= e^{i  \hat{H}_0 t}\mathfrak{P}(  \Phi(\; \cdot \; ) ) e^{-i  \hat{H}_0 t}
	\end{eqnarray*}
\end{enumerate}
\end{proof}

In the Nakajima–Zwanzig projector formalism  idempotents are usually called projectors in spite of the fact that they are not necessarily self-adjoint with respect to some scalar product.

Now let us take unitary dynamical map $ \Phi_t =  e^{i \hat{H} t} \; \cdot \; e^{-i \hat{H} t} $ as $ \Phi $.

\begin{proposition}
	For unitary dynamical map $ \Phi_t =  e^{i \hat{H} t} \; \cdot \; e^{-i \hat{H} t} $
	\begin{equation*}
		\mathfrak{P}(e^{i \hat{H} t} \; \cdot \; e^{-i \hat{H} t}) =\lim\limits_{T \rightarrow \infty} \frac{1}{T} \int_0^T ds  e^{i \hat{H}(s) t} \; \cdot \;  e^{-i \hat{H}(s) t} ,
	\end{equation*}
	where $ \hat{H}(s) = e^{-i \hat{H}_0 s} \hat{H}  e^{i \hat{H}_0 s}   $.
\end{proposition}

\begin{proof}
	By Eq. \eqref{eq:averProjDef}
	\begin{align*}
		\mathfrak{P}(e^{i \hat{H} t} \; \cdot \; e^{-i \hat{H} t}) =\lim\limits_{T \rightarrow \infty} \frac{1}{T} \int_0^T ds e^{-i \hat{H}_0 s} e^{i \hat{H} t} e^{i \hat{H}_0 s} \; \cdot \;e^{-i \hat{H}_0 s}  e^{-i \hat{H} t}  e^{i \hat{H}_0 s} =\\ =\lim\limits_{T \rightarrow \infty} \frac{1}{T} \int_0^T ds  e^{i \hat{H}(s) t} \; \cdot \;  e^{-i \hat{H}(s) t} .
	\end{align*}
\end{proof}

Thus, for unitary dynamics the effective Heisenberg dynamics is reduced to computation of free Heisenberg dynamics of Hamiltonian in the averaging parameter $ s $ and then Heisenberg dynamics with such a Hamiltonian in time $ t $. For products of operators it could be done for each operator separately.
\begin{corollary}\label{cor:averSevOp}
	Let $ X_1, \ldots, X_m  \in \mathcal{B}(\mathcal{H})$, then
	\begin{equation*}
		\mathfrak{P}(e^{i \hat{H} t} \; \cdot \; e^{-i \hat{H} t})(X_1 \ldots X_m) =\lim\limits_{T \rightarrow \infty} \frac{1}{T} \int_0^T ds  X_1(s;t) \ldots X_m(s;t),  
	\end{equation*}
	where $ X_k(s;t) \equiv e^{i \hat{H}(s) t} X_k e^{-i \hat{H}(s) t} $ for $ k= 1, \ldots,m $.
\end{corollary}

\section{Effective fermionic Heisenberg dynamics of moments}
\label{sec:effFermHeis}

Let us take $ \mathcal{H} = \mathbb{C}^{2^n}$. In such a space one could  (see Ref.~\cite{Takht11}, p. 407 for explicit formulae) define $ n $ pairs of fermionic creation and annihilation operators satisfying canonical anticommutation relations: $ \{\hat{c}_i^{\dagger}, \hat{c}_j \} = \delta_{ij},  \{\hat{c}_i, \hat{c}_j\} = 0 $.  We use the notation which is similar to Ref.~\cite{Ter17, Ter19, Ter19R}.  Let us define the $2n$-dimensional vector $\mathfrak{c} = (\hat{c}_1, \ldots, \hat{c}_n, \hat{c}_1^{\dagger}, \ldots, \hat{c}_n^{\dagger})^T$ of creation and annihilaiton operators. The  quadratic forms in such operators we denote by  $ \mathfrak{c}^T K \mathfrak{c} $, $ K \in \mathbb{C}^{2n \times 2n} $.  Define the $2n \times 2n$-dimensional matrix
\begin{equation*}
	E = \biggl(
	\begin{array}{cc}
		0 & I_n \\ 
		I_n & 0
	\end{array} 
	\biggr),
\end{equation*}
where $ I_n $ is the identity matrix from $ \mathbb{C}^{n \times n} $. Then canonical anticommutation relations take the form $ \{f^T\mathfrak{c}, \mathfrak{c}^Tg \} = f^TEg$, $ f, g \in \mathbb{C}^{2n}$. We also define the $\sim$-conju\-ga\-tion of matrices by the formula
\begin{equation*}
	\tilde{K} = E \overline{K} E, \qquad K \in \mathbb{C}^{2n \times 2n},
\end{equation*}
where the overline is an (elementwise) complex conjugation.

Now let us assume
\begin{equation*}
	 \hat{H} = \frac12\mathfrak{c}^T H \mathfrak{c}, \qquad \hat{H}_0 = \frac12\mathfrak{c}^T H_0 \mathfrak{c}
\end{equation*}
for $ \hat{H}  $ and $ \hat{H}_0 $ defined in the previous section, where $ H, H_0 \in \mathbb{C}^{2n \times 2 n}$ such that $  H = -H^T = -\tilde{H}  $, $  H_0 = -H_0^T = -\tilde{H}_0  $ (this conditions provide the self-adjointness $ \hat{H} =  \hat{H}^{\dagger} $, $ \hat{H}_0 =  \hat{H}_0^{\dagger} $ of such quadratic Hamiltonians without any other restrictions).

Let us formulate the special case of Lemma 1 from Ref.~\cite{Ter17} in the form which is useful for our case.
\begin{lemma}
	\label{lem:Hs}
	\begin{equation*}
		\hat{H}(s) \equiv e^{-i \hat{H}_0 s} \hat{H}  e^{i \hat{H}_0 s} =\frac12 \mathfrak{c}^T  H(s)\mathfrak{c},
	\end{equation*}
	where $ H(s) \equiv e^{-i H_0 E  s} H e^{i E H_0 s}  $.
\end{lemma}
And let us also give here Lemma 1 from Ref.~\cite{Nosal2020}. 
\begin{lemma}
	\label{lem:orthTransform} 
	Let $ H = -H^T \in \mathbb{C}^{2 n \times 2n} $, then  $ e^{ \frac{i}{2} \mathfrak{c}^T H \mathfrak{c} }  \mathfrak{c} e^{- \frac{i}{2} \mathfrak{c}^T H \mathfrak{c} } = O \mathfrak{c}$, where $O \equiv e^{-i E H} $.
\end{lemma}
They lead to the following lemma.
\begin{lemma}\label{lem:calcUnitTrans}
	\begin{equation*}
		 e^{i \hat{H}(s) t} \mathfrak{c} \otimes \ldots \otimes \mathfrak{c}   e^{-i \hat{H}(s) t} =( e^{-i E H(s) t} \otimes   \ldots \otimes  e^{-i E H(s) t} ) \mathfrak{c} \otimes \ldots \otimes \mathfrak{c}
	\end{equation*}
\end{lemma}
Here similar to Ref.~\cite{Nosal2020} $ e^{i \hat{H}(s) t} \mathfrak{c} \otimes \ldots \otimes \mathfrak{c}   e^{-i \hat{H}(s) t} $ means the tensor with elements $ e^{i \hat{H}(s) t} \mathfrak{c}_{i_1} \ldots \mathfrak{c}_{i_m}   e^{-i \hat{H}(s) t} $. 

Now let us prove the following proposition which allows one to reduce the computation of effective dynamics for moments to projections in the spaces whose dimension grows linear in $ n $ for fixed order moments.
\begin{proposition}\label{prop:effHeisDyn}
	\begin{equation*}
		\mathfrak{P}(e^{i \hat{H} t} \; \cdot \; e^{-i \hat{H} t})( \mathfrak{c} \otimes \ldots \otimes \mathfrak{c}  ) = (P^{(m)} (e^{-i H E  t} \otimes  \ldots \otimes  e^{-i H E t} ) ) \mathfrak{c} \otimes \ldots \otimes \mathfrak{c},
	\end{equation*}
	where
	\begin{equation}\label{eq:projOfmthOrder}
		P^{(m)} (X)= \lim\limits_{T \rightarrow \infty} \frac{1}{T} \int_0^T ds  (e^{-i H_0 E  s} \otimes  \ldots \otimes  e^{-i H_0 E  s} )  X  (e^{i H_0 E  s} \otimes  \ldots \otimes  e^{i H_0 E  s} ) 
	\end{equation}
	and $ m $ is the number of tensor factors in $ \mathfrak{c} \otimes \ldots \otimes \mathfrak{c}   $.
\end{proposition}

\begin{proof} 
	Since
	\begin{equation*}
		e^{-i E H(s)} =  e^{-i E  e^{-i H_0 E  s} H e^{i E H_0 s}} = e^{-i H_0 E  s} e^{-i E   H t } e^{i E H_0 s},
	\end{equation*}
	then we have
	\begin{eqnarray*}
		&e^{-i E H(s) t} \otimes   \ldots \otimes  e^{-i E H(s) t}  \\
		&=  (e^{-i H_0 E  s} \otimes  \ldots \otimes  e^{-i H_0 E  s} )   (e^{-i H E  t} \otimes  \ldots \otimes  e^{-i H E t} ) (e^{i H_0 E  s} \otimes  \ldots \otimes  e^{i H_0 E  s} ) 
	\end{eqnarray*}
	then by corollary \ref{cor:averSevOp} and lemma \ref{lem:calcUnitTrans} we have
	\begin{eqnarray*}
		&\mathfrak{P}(e^{i \hat{H} t} \; \cdot \; e^{-i \hat{H} t})( \mathfrak{c} \otimes \ldots \otimes \mathfrak{c}  ) \\
		&= \lim\limits_{T \rightarrow \infty} \frac{1}{T} \int_0^T ds ( e^{-i E H(s) t} \otimes   \ldots \otimes  e^{-i E H(s) t} ) \mathfrak{c} \otimes \ldots \otimes \mathfrak{c} \\
		&=  \lim\limits_{T \rightarrow \infty} \frac{1}{T} \int_0^T ds (e^{-i H_0 E  s} \otimes  \ldots \otimes  e^{-i H_0 E  s} ) \cdot \\
		&\cdot  (e^{-i H E  t} \otimes  \ldots \otimes  e^{-i H E t} ) (e^{i H_0 E  s} \otimes  \ldots \otimes  e^{i H_0 E  s} )  \mathfrak{c} \otimes \ldots \otimes \mathfrak{c}\\
		&=  (P^{(m)} (e^{-i H E  t} \otimes  \ldots \otimes  e^{-i H E t} ) ) \mathfrak{c} \otimes \ldots \otimes \mathfrak{c}
	\end{eqnarray*}
\end{proof}

If one additionally averages this effective Heisenberg dynamics with respect to the initial density matrix, then one obtains effective dynamics of $ m $-th order moments.

If one defines $ h^{(m)}_0 = -i \sum  I_{2n} \otimes \ldots I_{2n} \otimes H_0 E  \otimes I_{2n} \otimes \ldots I_{2n}$, then \eqref{eq:projOfmthOrder} takes the form
\begin{equation*}
	P^{(m)} (X)= \lim\limits_{T \rightarrow \infty} \frac{1}{T} \int_0^T ds  e^{h^{(m)}_0 s}  X e^{-h^{(m)}_0 s}
\end{equation*}
Similarly, if one defines  $ h^{(m)} = -i \sum  I_{2n} \otimes \ldots I_{2n} \otimes H E  \otimes I_{2n} \otimes \ldots I_{2n}$ 

\begin{equation*}
	\mathfrak{P}(e^{i \hat{H} t} \; \cdot \; e^{-i \hat{H} t})( \mathfrak{c} \otimes \ldots \otimes \mathfrak{c}  ) = (P^{(m)} (e^{h^{(m)} t} ) \mathfrak{c} \otimes \ldots \otimes \mathfrak{c}.
\end{equation*}

Thus, the computation of the effective Heisenberg dynamics of arbitrary order  moments for creation and annihilation operators  is reduced to similar  computations of
\begin{equation*}
	P (e^{ht} ), \qquad 	P (\; \cdot \; ) \equiv  \lim\limits_{T \rightarrow \infty} \frac{1}{T} \int_0^T ds  e^{h_0 s}  \; \cdot \; e^{-h_0 s},
\end{equation*}
for correspondent  $ h^{(m)} $ and $ h_0^{(m)} $ as $ h $ and $ h_0 $.

The map $ P $  inherits the properties of the map $ \mathfrak{P} $ from Prop. \ref{prop:propOfP}.
\begin{proposition}
	\begin{enumerate}
		\item $ P $ is idempotent.
		\begin{equation*}
			P^2 = P
		\end{equation*}
		\item The result of application of $ P $  to an arbitrary matrix $ X $ commutes with $ e^{-h_0 t} $
		\begin{equation}\label{eq:commWithP}
			e^{-h_0 t}	P ( X ) =  P ( X ) e^{-h_0 t}
		\end{equation}
		\item Pulling $ e^{-h_0 t} $ through the map $ P $ 
		\begin{equation*}
			P (e^{-h_0 t} X ) = e^{-h_0 t} P ( X )
		\end{equation*}
		for an arbitrary matrix $ X $. 
	\end{enumerate}
\end{proposition}

\section{Time-local effective fermionic Heisenberg equations}
\label{sec:timeLocalFermHeis}

Now let us consider $ H = H_0 + \lambda H_I $, where $ \lambda \rightarrow 0 $ is a small. Then it is possible to represent $ h^{(m)} $ as
\begin{equation*}
	h^{(m)} = h_0^{(m)} + \lambda h_I^{(m)},
\end{equation*}
where $ h^{(m)}_I = -i \sum  I_{2n} \otimes \ldots I_{2n} \otimes H_I E  \otimes I_{2n} \otimes \ldots I_{2n}$.

So  as at the end of the previous section let us firstly discuss the general case omitting the index referring to the specific $ m $. So let we have
$ h = h_0 + \lambda h_I $.

As we are interested in an asymptotic expansion in $ \lambda $, so let us turn to ''interaction'' representation.  
\begin{proposition}
	Let $ v(t) \equiv e^{-h_0t} e^{ht} $, then
	\begin{equation*}
		\frac{d}{dt} v(t) = \lambda h_I(t)  v(t),
	\end{equation*}
	where $ h_I(t) \equiv e^{-h_0t} h_I e^{h_0t}$.
\end{proposition}

\begin{proof}
	By direct calculation we have
	\begin{eqnarray*}
		\frac{d}{dt} v(t) = \frac{d}{dt} (e^{-h_0t} e^{ht}) 
		= e^{-h_0t} (- h_0+h_0 +\lambda h_I e^{ht}) \\
		= \lambda e^{-h_0t}  h_I e^{h_0t} (e^{-h_0t}e^{ht})  = \lambda h_I(t)  v(t).
	\end{eqnarray*}	
\end{proof}

Then the asymptotic expansion of $ v(t) $ could be calculated by the Dyson series.
\begin{corollary}
	The asymptotic expansion of $ P v(t)  $ has the form
	\begin{equation}\label{eq:PvTexpSeries}
		P v(t) = I + \sum_{k=1}^{\infty} \lambda^k \mu_k(t),
	\end{equation}
	where
	\begin{equation}\label{eq:mukDef}
		\mu_k(t) = \int_0^t dt_k \ldots \int_0^{t_2} dt_1 P(h_I(t_k) \ldots h_I(t_1)).
	\end{equation}
\end{corollary}

Let us note the following ''stationarity'' property.
\begin{proposition}
	\begin{equation*}
		P(h_I(t_k) \ldots h_I(t_1)) = P(h_I(t_k- t) \ldots h_I(t_1 - t)) 
	\end{equation*} 
\end{proposition}
\begin{proof}
	Taking into account Eq. \eqref{eq:commWithP} we obtain
	\begin{eqnarray*}
		&P(h_I(t_k) \ldots h_I(t_1)) = e^{- h_0 t} P(h_I(t_k) \ldots h_I(t_1)) e^{h_0 t} \\
		&= \lim\limits_{T \rightarrow \infty} \frac{1}{T} \int_0^T ds    e^{h_0 s}  e^{- h_0 t}  h_I(t_k) \ldots h_I(t_1) e^{-h_0 s} \\
		&= \lim\limits_{T \rightarrow \infty} \frac{1}{T} \int_0^T ds  e^{h_0 s}   e^{-h_0 t_k} h_I e^{h_0 t_k} e^{- h_0 t} e^{ h_0 t}   \ldots e^{- h_0 t} e^{ h_0 t}  e^{-h_0 t_1} h_I e^{h_0 t_1} e^{h_0 t}  e^{-h_0 s} 
		\\
		&= \lim\limits_{T \rightarrow \infty} \frac{1}{T} \int_0^T ds  e^{h_0 s}  h_I(t_k - t) \ldots h_I(t_1 -t)e^{-h_0 s} = P(h_I(t_k- t) \ldots h_I(t_1 - t)) 
	\end{eqnarray*}
\end{proof}

In particular, it leads to
\begin{equation*}
	P(h_I(t_1)) = P(h_I), \qquad  P(h_I(t_2) h_I(t_1))  = P(h_I h_I(t_1-t_2)) 
\end{equation*}

This allows one to calculate first terms of expansion in Eq.~\eqref{eq:PvTexpSeries}.
\begin{equation*}
	\mu_1(t) =  \int_0^{t} dt_1 P(h_I(t_1)) =  \int_0^{t} dt_1 P(h_I)= t P(h_I)
\end{equation*}
\begin{eqnarray*}
	\mu_2(t) = \int_0^t dt_2 \int_0^{t_2} dt_1 P(h_I(t_2) h_I(t_1)) = \int_0^t dt_2 \int_0^{t_2} dt_1 P(h_I h_I(t_1-t_2))  \\
	= \int_0^t dt_2 \int_0^{t_2} dt_1 P(h_I h_I(-t_1)) = \int_0^t dt_2 \int_0^{t_2} dt_1 P(h_I e^{t_1 [h_0, \; \cdot \;]} h_I) \\
	=  P\left(h_I \frac{e^{t [h_0, \; \cdot \;]} -1 -  [h_0, \; \cdot \;] t }{[h_0, \; \cdot \;]^2} h_I\right)
\end{eqnarray*}

Here the function of $ [h_0, \; \cdot \;] $ is defined by the series 
\begin{equation*}
	\frac{e^{t [h_0, \; \cdot \;]} -1 -  [h_0, \; \cdot \;] t }{[h_0, \; \cdot \;]^2} \equiv \sum_{k=2}^{\infty} ([h_0, \; \cdot \;])^{k-2} \frac{t^k}{k!},
\end{equation*}
so it is well-defined  in spite of degeneracy of $  [h_0, \; \cdot \;] $.

Now let us perturbatively obtain a time-local equation for $ P v(t)  $.
\begin{proposition}
	The generator $  l_I(t) $ of the time-local equation for $ P v(t) $
	\begin{equation*}
		\frac{d}{dt} P v(t) = l_I(t) P v(t) 
	\end{equation*}
	has asymptotic expansion
	\begin{equation*}
		l_I(t) = \sum_k \lambda^k \kappa_k(t),
	\end{equation*}
	where
	\begin{equation}\label{eq:kappaDef}
		\kappa_k(t) =  \sum_{\sum k_j = k} (-1)^q \dot{\mu}_{k_0}(t) \mu_{k_1}(t) \ldots \mu_{k_q}(t)
	\end{equation}
\end{proposition}

\begin{proof}
	By direct computation we have
	\begin{eqnarray*}
		\left(\frac{d}{dt} P v(t)\right)(P v(t) )^{-1} =  \sum_{k=1}^{\infty} \lambda^k \dot{\mu}_k(t) \left(I + \sum_{k=1}^{\infty} \lambda^k \mu_k(t)\right)^{-1} \\
		=  \sum_{k=1}^{\infty}  \lambda^k \sum_{\sum k_j = k} (-1)^q \dot{\mu}_{k_0}(t) \mu_{k_1}(t) \ldots \mu_{k_q}(t).
	\end{eqnarray*}
\end{proof}
This is nothing else but a variant of Kubo-van Kampen cumulant expansion\cite{Kubo1963, vanKampen1974, vanKampen1974a}, which is widely used to  obtain convolutionless master equations\cite{Chaturvedi79}. But usually time-integrals coming from Eq. \eqref{eq:mukDef} are rearranged in Eq. \eqref{eq:kappaDef} in such a way, that they better converge in long-time limit. However, this convergence only occurs for the systems with infinite degrees of freedom and, thus, is not relevant for purposes of the present paper.

The first two terms have the form
\begin{equation*}
	\kappa_1(t) = \dot{\mu}_1(t) =  P(h_I),
\end{equation*}
\begin{equation*}
	\kappa_2(t) = \dot{\mu}_2(t) - \dot{\mu}_1(t) \mu_1(t) =  P\left(h_I \frac{e^{t [h_0, \; \cdot \;]} -1  }{[h_0, \; \cdot \;]} h_I\right) -  t( P(h_I))^2.
\end{equation*}

If one needs to turn back from the representation picture, then one could just add the free generator $ h_I $. Namely, one has
\begin{equation*}
	\frac{d}{dt} P e^{ht} = l(t) P e^{ht}, 
\end{equation*}
where
\begin{equation*}
	l(t) \equiv h_0 +  e^{h_0 t}l_I(t) e^{-h_0 t}  = h_0 + l_I(t) .
\end{equation*}

Thus, the one has asymptotic expansion
\begin{equation}\label{eq:ltFirstTemrs}
	l(t) = h_0 + \lambda  P(h_I)  + \lambda^2 \left( P\left(h_I \frac{e^{t [h_0, \; \cdot \;]} -1  }{[h_0, \; \cdot \;]} h_I\right) -  t( P(h_I))^2\right) + O(\lambda^3).
\end{equation}

This leads to effective time-local Heisenberg equations
\begin{equation*}
	\frac{d}{dt}(\mathfrak{c}  \otimes \ldots \otimes \mathfrak{c})(t) = l^{(m)}(t) (\mathfrak{c}  \otimes \ldots \otimes \mathfrak{c})(t) ,
\end{equation*}
where one should restore index $ (m) $ for $ h_I $ and $ h_0 $ in Eq. \eqref{eq:ltFirstTemrs} and use their definitions in terms of $ H $ and $ H_0 $.

\section{Bosonic case at physical level of rigor and some caveats}
\label{sec:bosonicCase}

Let us briefly consider a formal bosonic analog of the above results. It is important due to the fact that it is still very actively studied area\cite{Agredo21, Agredo21a, Gaidash21, Medvedeva21, Volovich21}. We use the notations similar to Refs.~\cite{Ter16, Ter17a, Ter19}. We consider the Hilbert space $\otimes_{j=1}^n\ell_2$.  In such a space one could (see e.g. Ref.~\cite{scalli2003}, Paragraph 1.1.2) define $n$ pairs of creation and annihilation operators satisfying canonical commutation relations: $ [\hat{a}_i, \hat{a}_j^{\dagger}] = \delta_{ij}$, $ [\hat{a}_i, \hat{a}_j] = [\hat{a}_i^{\dagger}, \hat{a}_j^{\dagger}]= 0 $.  Let us define the $2n$-dimensional vector $\mathfrak{a} = (\hat{a}_1, \cdots, \hat{a}_n, \hat{a}_1^{\dagger}, \cdots, \hat{a}_n^{\dagger} )^T$.  Linear and quadratic forms in such operators we denote by $ f^T \mathfrak{a}  $ and $ \mathfrak{a}^T K \mathfrak{a} $, respectively.  Here, $  f \in \mathbb{C}^{2n} $ and $ K \in \mathbb{C}^{2n \times 2n} $.  Define the $2n \times 2n$-dimensional matrices as
\begin{equation*}
	J = \biggl(
	\begin{array}{cc}
		0 & -I_n \\ 
		I_n & 0
	\end{array} 
	\biggr).
\end{equation*}
Canonical commutation relations in such a notation take the form $ [f^T \mathfrak{a},\mathfrak{a}^T g]  = - f^T J g,  \forall g, f \in \mathbb{C}^{2n} $.

Similar to the fermionic case let us assume $ \hat{H}_0 = \frac12 \mathfrak{a}^T H_0 \mathfrak{a} $ and $ \hat{H} = \frac12 \mathfrak{a}^T H \mathfrak{a} $ in Cor.~\ref{cor:averSevOp} with $ H, H_0 \in \mathbb{C}^{2n \times 2n} $ such that $ H = H^T = \tilde{H} $ , $ H_0 = H_0^T = \tilde{H}_0 $ (which provides the self-adjointness $  \hat{H} =  \hat{H}^{\dagger}  $, $  \hat{H}_0 =  \hat{H}_0^{\dagger}  $). Then similar to Lemma \ref{lem:Hs}  due to Lemma 3.1 from Ref.~\cite{Ter17a}, we have
\begin{eqnarray*}
	\hat{H}(s) \equiv e^{-i \hat{H}_0 s} \hat{H}  e^{i \hat{H}_0 s}  = e^{-i \frac12 \mathfrak{a}^T H_0 \mathfrak{a}s} \frac12 \mathfrak{a}^T H \mathfrak{a}  e^{i \frac12 \mathfrak{a}^T H_0 \mathfrak{a} s} \\
	= \frac12 \mathfrak{a}^T e^{-i H_0 J s} H  e^{i J H_0 s} \mathfrak{a} = \frac12 \mathfrak{a}^T H(s) \mathfrak{a},
\end{eqnarray*}
where $ H(s) \equiv e^{-i H_0 J s} H  e^{i J H_0 s} $.

Then similar to Prop.~\ref{prop:effHeisDyn} we have
\begin{equation*}
	\mathfrak{P}(e^{i \hat{H} t} \; \cdot \; e^{-i \hat{H} t})( \mathfrak{a} \otimes \ldots \otimes \mathfrak{a}  ) = (P^{(m)} (e^{i H J t} \otimes  \ldots \otimes  e^{-i H J t} ) ) \mathfrak{a} \otimes \ldots \otimes \mathfrak{a},
\end{equation*}
where
\begin{equation}
	P^{(m)} (X)= \lim\limits_{T \rightarrow \infty} \frac{1}{T} \int_0^T ds  (e^{-i H_0 J  s} \otimes  \ldots \otimes  e^{-i H_0 J  s} )  X  (e^{i H_0 J  s} \otimes  \ldots \otimes  e^{i H_0 J  s} ). 
\end{equation}

In particular, we have  
\begin{equation*}
	P^{(1)} (X) = \lim\limits_{T \rightarrow \infty} \frac{1}{T} \int_0^T ds  e^{-i  H_0  J s} X  e^{i H_0  J s} .
\end{equation*}

Nevertheless, it is well-know (see e.g. Ref.~\cite{Arnold89}, App.~6) that $   H_0 J  $ can have non-real eigenvalues, thus the limit at the right-hand side of this formula will have exponentially diverging terms.

\section{Conclusions}
\label{sec:Concl}

We have defined the averaging projector for quantum dynamical maps, which turns them to effective ones.  We have illustrated our approach by the case of fermionic dynamics with a quadratic generator. In particular, we have perturbatively obtained the generators of  time-local effective Heisenberg equations for arbitrary numbers of multiplied creation and annhilation operators.

We have shown that the formal bosonic analog of such results leads to contradiction. So it seems that one should reformulate the initial setup for infinite-dimensional Hilbert space differently. Possibly it could be done in way similar to that for master equations (see e.g. Ref.~\cite{Holevo03}, Sec.3.3.3). Anyway it is an important direction of the further development. 

Another important question for future study is how the effective dynamics occurring in this work is related to the one described by effective Hamiltonians\cite{Soliverez81, Basharov21}. Actually,  it seems that zeroth and first order terms of our expansion will coincide with dynamics in the rotating wave approximation and the same is true for the  effective Hamiltonian case\cite{Thimmel99, Trubilko20}. Nevertheless, we doubt that it is possible to obtain our second order terms of our expansion from any effective Hamiltonian.

\end{document}